\documentclass[aps, pra, superscriptaddress, nofootinbib, showpacs]{revtex4-1}
\usepackage{amsfonts,amssymb,amsmath}
\usepackage{amsthm}
\usepackage{graphics,graphicx,epsfig}
\usepackage[dvipsnames]{xcolor}
\usepackage[normalem]{ulem}
\usepackage{comment}
\usepackage{epstopdf}
\usepackage{subfigure}
\usepackage{float}

\usepackage{hyperref}

\newtheorem{theorem}{Theorem}

\def\be{\begin{eqnarray}}
\def\ee{\end{eqnarray}}
\def\Tr{\mathrm{Tr}}
\begin{document}

\title{Entanglement criterion via general symmetric informationally complete measurements}
\author{Le-Min Lai}
\affiliation{School of Mathematical Sciences,  Capital Normal University,  Beijing 100048,  China}
\author{Tao Li}
\affiliation{School of Science, Beijing Technology and Business University, Beijing 100048, China}
\author{Shao-Ming Fei}
\affiliation{School of Mathematical Sciences,  Capital Normal University,  Beijing 100048,  China}
\affiliation{Max Planck Institute for Mathematics in the Sciences, 04103, Leipzig, Germany}
\author{Zhi-Xi Wang}
\affiliation{School of Mathematical Sciences,   Capital Normal University,   Beijing 100048,   China}

\begin{abstract}
We study the quantum separability problem by using general symmetric informationally complete measurements and present a separability criterion for arbitrary dimensional bipartite systems.
We show by detailed examples that our criterion is more powerful than the existing ones in entanglement
detection.
\end{abstract}

\maketitle

\section{Introduction}
\label{intro}
Quantum entanglement is one of the most fundamental resources in quantum information processing \cite{1, 2, 3}. Operational and efficient criteria for the detection of entanglement are of great significance. It has been discussed that the problem of determining whether or not a given state is entangled is NP-hard \cite{4, 5, 6}. There have been numerous criteria to distinguish quantum entangled states from the separable ones,  such as positive partial transposition(PPT) criterion \cite{7, 8, 9}, realignment criterion \cite{10, 11, 12, 13, 14, 15, 16}, covariance matrix criterion \cite{17}, correlation matrix criterion \cite{18, 19} and generalized form of the correlation matrix criterion \cite{20}.

While numerous mathematical tools have been employed  in entanglement detection of  quantum states,  experimental implementation of entanglement detection for unknown quantum states has fewer results \cite{21, 22, 23, 24}.  In \cite{25}, the authors connected the separability problem with the concept of mutually unbiased bases (MUBs) \cite{26} for two-qubit,  multipartite and continuous-variable quantum systems. These entanglement criteria are shown to be powerful and can be implemented experimentally.  After that, the authors in \cite{27, 28} generalized such idea and provided an entanglement criterion based on mutually unbiased measurements (MUMs) \cite{29}. Moreover, it has been shown that the criterion based on MUMs is more effective than the criterion based on MUBs.

Besides mutually unbiased bases, another intriguing topic in quantum information theory is the symmetric informationally complete positive operator-valued measures (SIC-POVMs). Most of the literature on SIC-POVMs focus on rank-$1$ SIC-POVMs that all the positive operator-valued measure (POVM) elements are proportional to rank-$1$ projectors. Nevertheless, the existence of SIC-POVMs in arbitrary dimension is still an open problem \cite{30}. In Appleby \cite{31}, the author introduced the general symmetric informationally complete measurements (general SIC-POVMs) in which the elements need not to be rank one, and showed that general SIC-POVMs exist in all finite dimensions. Furthermore, Gour and Kalev \cite{32} constructed all general SIC measurements from the generalized Gell-Mann matrices. In Chen et al. and Xi et al. \cite{33, 34}, the authors  presented separability criteria for both $d$-dimensional bipartite and multipartite systems based on general SIC-POVMs. Very recently, Bae et al. \cite{35} studied entanglement detection via quantum $2$-designs, which includes SIC-POVMs as a special example. In Czartowki et al.\cite{36}, the authors investigated the entanglement properties of multipartite systems with tight informationally complete measurements including SIC-POVMs. In addition, the authors in \cite{37} considered a nonlinear entanglement criterion based on SIC-POVMs. In Shang et al.\cite{38}, the authors used the SIC-POVMs to derive the entanglement criterion  and demonstrated the superiority of the criterion by various examples.

In this paper, we study the quantum separability problem by using general SIC-POVMs and present a separability criterion  for arbitrary high dimensional bipartite systems of a $d_A$-dimensional subsystem and a $d_B$-dimensional subsystem.  The paper is organized as follows: In  Sect.~\ref{sec:1}, we recall some basic notions of  SIC-POVMs and general SIC-POVMs.  Section \ref{sec:2} provides an entanglement  criterion based on the general SIC-POVMs and some remarks.  In Sect.~\ref{sec:3},  we compare the criterion with the ones in \cite{33} and \cite{38} via detailed examples, and show that our criterion is more efficient than the existing ones. We conclude the paper in Sect.~\ref{sec:4}.
\section{SIC-POVMs and general SIC-POVMs}
\label{sec:1}
 We first review some basic knowledge of symmetric informationally complete measurements  and general symmetric informationally complete measurements. A POVM  $\{P_j\}$ with $d^2$ rank-$1$ operators acting on $\mathbb{C}^d$ is symmetric informationally complete, if
\begin{eqnarray}
&&P_j=\frac{1}{d}|\phi_j\rangle\langle\phi_j|, \\
&&\sum_{j =1}^{d^2}P_j=\mathbb{I},
\end{eqnarray}
where $j=1, 2, \dots, d^2$, $\mathbb{I}$ is the identity operator, and the vectors $|\phi_j\rangle$ satisfy $|\langle\phi_j|\phi_k\rangle|^2={1}/({d+1})$, $j\neq k$. The existence of SIC-POVMs in arbitrary dimension $d$ is still an open problem.  Only analytical solutions have been found in dimensions $d = 2-24,  28, 30,  31,  35,  37,  39,  43,  48,  124$, and  numerical solutions have been found up to dimension $151$ \cite{30}.

The concept and constructions of general SIC measurements are introduced in Refs. \cite{31, 32}. A set of $d^2$ positive semidefinite operators $\{P_\alpha\}_{\alpha =1}^{d^2}$ on $\mathbb{C}^{d}$ is said to be a general SIC measurements if
\begin{eqnarray}
&&\sum_{\alpha =1}^{d^2}P_\alpha=\mathbb{I}, \\
&&\Tr(P_\alpha^2)=a, \\
&&\Tr(P_\alpha P_\beta)=\frac{1-da}{d(d^2-1)},
\end{eqnarray}
where $\alpha, \beta \in \{1, 2, \dots, d^2\}, \alpha \neq \beta$, the parameter $a$ satisfies $\frac{1}{d^3} <a \leqslant \frac{1}{d^2}$, and $a=\frac{1}{d^2}$ if and only if all $P_\alpha$ are rank one, which gives rise to a SIC-POVM. It can be shown that $\Tr(P_\alpha)=\frac{1}{d}$ for all $\alpha$. Contrasting to SIC-POVM, the general SIC-POVM can be explicitly constructed \cite{32}. Let $\{F_\alpha\}_{\alpha=1}^{d^2-1}$ be a set of $d^2-1$ Hermitian, traceless operators acting on $\mathbb{C}^d$, satisfying $\Tr(F_\alpha F_\beta)=\delta_{\alpha, \beta}$. Define $F=\sum_{\alpha=1}^{d^2-1}F_\alpha$. The $d^2$ operators
\be
&&P_\alpha=\frac{1}{d^2}\mathbb{I}+t[F-d(d+1)F_\alpha], \alpha=1, 2, \dots, d^2-1,\\
&&P_{d^2}=\frac{1}{d^2}\mathbb{I}+t(d+1)F
\ee
form a general SIC measurement. Here $t$ should be chosen such that $P_\alpha\geqslant0$ and the parameter $a$ is given by
\be
a=\frac{1}{d^3}+t^2(d-1)(d+1)^3.
\ee

\section{Entanglement detection via general SIC-POVMs}\label{sec:2}
Entanglement detection via SIC-POVMs had been discussed in \cite{38}. However, the method subjects to the existence of SIC-POVMs, which is an open problem. Unlike the SIC-POVMs, the general symmetric informationally complete measurements do exist for arbitrary dimension $d$.

Consider a quantum state $\rho$ and a general SIC-POVM $\mathcal{M}_s=\{P_\alpha\}_{\alpha =1}^{d^2}$. We have the probability $p_\alpha=\langle P_\alpha \rangle=\Tr(P_\alpha \rho)$ of outcome $\alpha$. Conversely,  the quantum state $\rho$ can be reconstructed from these probabilities:
\be
\rho=\frac{d(d^2-1)}{ad^3-1}\sum_{\alpha=1}^{d^2}p_\alpha P_\alpha-\frac{d(1-ad)}{ad^3-1}\mathbb{I}.
\ee
Denote $(e|=(p_1\  p_2 \ \cdots\  p_{d^2})$ and $|e)=(p_1\  p_2 \ \cdots\  p_{d^2})^\mathrm{T}$.
Calculation shows that
\be
\sum_{\alpha=1}^{d^2}p_\alpha^2&=&\frac{(ad^3-1)\Tr(\rho^2)+d(1-ad)}{d(d^2-1)}\nonumber\\
&\leqslant& \frac{ad^2+1}{d(d+1)},
\ee
where the upper bound is saturated iff $\rho$ is pure.

Now consider a $d_A\times d_B$ bipartite state $\rho$, and
two general SIC-POVMs: $\{P_\alpha^A\}_
{\alpha=1}^{d_A^2}$ with parameter $a_A$ and $\{P_\alpha^B\}_{\alpha=1}^{d_B^2}$ with parameter $a_B$ for the two subsystems, respectively. The linear correlations between $P^A$ and $P^B$ read
\be
[\mathcal{P}]_{ij}=\langle P_i^A\otimes P_j^B\rangle.
\ee
Denote $\mathcal{P}$ the matrix with entries given by $[\mathcal{P}]_{ij}$.

\begin{theorem} If a bipartite state $\rho$ is separable, then
\be
\|\mathcal{P}\|_{\mathbf{tr}}\leqslant \sqrt{\frac{a_Ad_A^2+1}{d_A(d_A+1)}}\sqrt{\frac{a_Bd_B^2+1}{d_B(d_B+1)}},
\ee
where $\|\mathcal{P}\|_{\mathbf{tr}}=\Tr(\sqrt{ \mathcal{P}\mathcal{P}^\dagger})$.

\end{theorem}
\begin{proof}
We consider a pure separable state $\rho=\rho_A \otimes\rho_B$ at first. We have
\be
\mathcal{P}&=&\left(
  \begin{array}{c}
    p_{A, 1} \\
    \vdots \\
    p_{A, d_A^2} \\
  \end{array}
\right)
\left(
  \begin{array}{ccc}
    p_{B, 1} & \cdots & p_{B, d_B^2} \\
  \end{array}
\right)\equiv |e_A)(e_B|,
\ee
where $p_{A, i}=\Tr(P_i^A\rho)$ for $i=1, 2, \dots, d_A^2$ and $p_{B, j}=\Tr(P_j^B\rho)$ for $j=1, 2, \dots, d_B^2$. Then
\be
\| \mathcal{P}\|_{\mathbf{tr}}&=&(e_A\mid e_A)^{\frac{1}{2}}(e_B\mid e_B)^{\frac{1}{2}} \nonumber\\
&\leqslant& \sqrt{\frac{a_Ad_A^2+1}{d_A(d_A+1)}}\sqrt{\frac{a_Bd_B^2+1}{d_B(d_B+1)}}.
\ee
By employing the convexity property of the trace norm, we have
\be
\| \mathcal{P}\|_{\mathbf{tr}}\leqslant \sqrt{\frac{a_Ad_A^2+1}{d_A(d_A+1)}}\sqrt{\frac{a_Bd_B^2+1}{d_B(d_B+1)}}
\ee
for separable states.
\begin{flushright}
\end{flushright}
\end{proof}

\textit{\textbf{Remark 1.}} If one takes $a=\frac{1}{d^2}$, the criterion of Theorem 1 reduces to the criterion based on SIC-POVM \cite{38}, i.e., if a bipartite state $\rho$ is separable, then $\| \mathcal{P}\|_{\mathbf{tr}}\leqslant \sqrt{\frac{2}{d_A(d_A+1)}}\sqrt{\frac{2}{d_B(d_B+1)}}$.

\textit{\textbf{Remark 2.}} If $d_A=d_B$ and $a_A=a_B$, we have $\|\mathcal{P}\|_{\mathbf{tr}}\leqslant\frac{ad^2+1}{d(d+1)}$. Furthermore, for a product state, one gets
\be
J_a(\rho)\leqslant \| \mathcal{P}_s\|_{\mathbf{tr}},
\ee
where $J_a(\rho)=\sum\limits_{j=1}\limits^{d^2}\Tr(P_j\otimes Q_j \rho)$ \cite{33}.

\section{Examples}\label{sec:3}

Let us consider some examples to illustrate the effectiveness and superiority  of our criterion compared with the previously known criterion in \cite{33}  and the recently criterion in \cite{38}.

Let $\{P_\alpha\}_{\alpha=1}^{d^2}$ be a set of general SIC-POVM on $\mathbb{C}^d$ with the parameter $a$. Let $\bar{P}_\alpha$ denote the conjugation of $P_\alpha$. Then $\{\bar{P}_\alpha\}_{\alpha=1}^{d^2}$ is another set of general SIC-POVM with the same parameter $a$. We consider the case of $d=3$.

It can be shown that for any nonzero $t\in[-0. 012, 0. 012]$, the following nine matrices
\be
&&P_\alpha=\frac{1}{9}\mathbb{I}+t(G_9-12G_\alpha), \ \mathrm{for}\  \alpha=1, 2, \dots, 8,\\
&&P_9=\frac{1}{9}\mathbb{I}+4tG_9
\ee
form a general SIC-POVM, where
\begin{eqnarray*}
G_1=\left(
\begin{array}{ccc}
 \frac{1}{\sqrt{2}} & 0 & 0 \\
 0 & -\frac{1}{\sqrt{2}} & 0 \\
 0 & 0 & 0 \\
\end{array}
\right),~~
G_2=\left(
\begin{array}{ccc}
 0 & \frac{1}{\sqrt{2}} & 0 \\
 \frac{1}{\sqrt{2}} & 0 & 0 \\
 0 & 0 & 0 \\
\end{array}
\right),~~
G_3=\left(
\begin{array}{ccc}
 0 & 0 & \frac{1}{\sqrt{2}} \\
 0 & 0 & 0 \\
 \frac{1}{\sqrt{2}} & 0 & 0 \\
\end{array}
\right),
\end{eqnarray*}
\begin{eqnarray*}
G_4=\left(
\begin{array}{ccc}
 0 & -\frac{i}{\sqrt{2}} & 0 \\
 \frac{i}{\sqrt{2}} & 0 & 0 \\
 0 & 0 & 0 \\
\end{array}
\right),~~
G_5=\left(
\begin{array}{ccc}
 \frac{1}{\sqrt{6}} & 0 & 0 \\
 0 & \frac{1}{\sqrt{6}} & 0 \\
 0 & 0 & -\sqrt{\frac{2}{3}} \\
\end{array}
\right),~~
G_6=\left(
\begin{array}{ccc}
 0 & 0 & 0 \\
 0 & 0 & \frac{1}{\sqrt{2}} \\
 0 & \frac{1}{\sqrt{2}} & 0 \\
\end{array}
\right),
\end{eqnarray*}
\begin{eqnarray*}
G_7=\left(
\begin{array}{ccc}
 0 & 0 & -\frac{i}{\sqrt{2}} \\
 0 & 0 & 0 \\
 \frac{i}{\sqrt{2}} & 0 & 0 \\
\end{array}
\right),~~
G_8=\left(
\begin{array}{ccc}
 0 & 0 & 0 \\
 0 & 0 & -\frac{i}{\sqrt{2}} \\
 0 & \frac{i}{\sqrt{2}} & 0 \\
\end{array}
\right),~~
G_9=\left(
\begin{array}{ccc}
 \frac{1}{\sqrt{2}}+\frac{1}{\sqrt{6}} & \frac{1-i}{\sqrt{2}} & \frac{1-i}{\sqrt{2}} \\
 \frac{1+i}{\sqrt{2}} & -\frac{1}{\sqrt{2}}+\frac{1}{\sqrt{6}} & \frac{1-i}{\sqrt{2}} \\
 \frac{1+i}{\sqrt{2}} & \frac{1+i}{\sqrt{2}} & -\sqrt{\frac{2}{3}} \\
\end{array}
\right).
\end{eqnarray*}
We can use the two general SIC-POVMs  $\{P_\alpha\}_{\alpha=1}^{9}$ and  $\{\bar{P}_\alpha\}_{\alpha=1}^{9}$  to recognize entanglement.

\textit{\textbf{Example 1.}} Consider the isotropic states that are locally unitarily equivalent to a maximally entangled state mixed with white noise:
\be
\rho_{\mathrm{iso}}=q\mid\Phi^{+}\rangle\langle\Phi^{+}|+(1-q)\frac{\mathbb{I}}{d^2},\ \ \  0\leqslant q\leqslant1,
\ee
where $\mid\Phi^{+}\rangle=\displaystyle\frac{1}{\sqrt{d}}\sum\limits_{i=0}\limits^{d-1}\mid ii\rangle$.
For $d=3$, by directly calculating the correlation entries $[\mathcal{P}]_{ij}=\langle P_i\otimes \bar{P}_j\rangle$, $i,j=1,\dots,9$, we have
$\|\mathcal{P}\|_{\mathbf{tr}}-\frac{9a+1}{12}=96t^2(4q-1)>0$ for $\frac{1}{4}< q\leqslant 1$. Thus, our criterion can detect the entanglement of the state $\rho_{\mathrm{iso}}$ for $\frac{1}{4}< q\leqslant 1$.

\textit{\textbf{Example 2.}} Consider the Werner states \cite{39}
\be
W_d\equiv \frac{1}{d^3-d}((d-f)\mathbb{I}_{d^2}+(df-1)V),
\ee
where $-1\leqslant f\leqslant 1$, $V=\sum_{i,j=0}^{d-1}|ij\rangle \langle ji|$. $W_d$ is entangled if and only if $-1\leqslant f<0$.
For $d=3$, by direct calculation we have $\|\mathcal{P}\|_{\mathbf{tr}}-\frac{9a+1}{12}=48t^2(\sqrt{(3f-1)^2}-2)>0$ for $-1\leq f <-\frac{1}{3}$. Thus our criterion recognizes the entanglement for $-1\leqslant f <-\frac{1}{3}$.
From the criterion in \cite{33}, one has $J_a(W_3)-\frac{9a+1}{12}=\sum\limits_{j=1}\limits^{d^2}\Tr(P_j\otimes \bar{P}_j W_3)-\frac{9a+1}{12}=36(f-3)t^2<0$, since $-1\leqslant f\leqslant 1$. Hence, our criterion is more efficient than the criterion in \cite{33}.

\textit{\textbf{Example 3.}} Consider the $3\times 3$ bound entangled state $\rho^x$ \cite{9},
\begin{center}
\be
\rho^x=\left(
\begin{array}{ccccccccc}
 \frac{x}{8 x+1} & 0 & 0 & 0 & \frac{x}{8 x+1} & 0 & 0 & 0 & \frac{x}{8 x+1} \\
 0 & \frac{x}{8 x+1} & 0 & 0 & 0 & 0 & 0 & 0 & 0 \\
 0 & 0 & \frac{x}{8 x+1} & 0 & 0 & 0 & 0 & 0 & 0 \\
 0 & 0 & 0 & \frac{x}{8 x+1} & 0 & 0 & 0 & 0 & 0 \\
 \frac{x}{8 x+1} & 0 & 0 & 0 & \frac{x}{8 x+1} & 0 & 0 & 0 & \frac{x}{8 x+1} \\
 0 & 0 & 0 & 0 & 0 & \frac{x}{8 x+1} & 0 & 0 & 0 \\
 0 & 0 & 0 & 0 & 0 & 0 & \frac{x+1}{2 (8 x+1)} & 0 & \frac{\sqrt{1-x^2}}{2 (8 x+1)} \\
 0 & 0 & 0 & 0 & 0 & 0 & 0 & \frac{x}{8 x+1} & 0 \\
 \frac{x}{8 x+1} & 0 & 0 & 0 & \frac{x}{8 x+1} & 0 & \frac{\sqrt{1-x^2}}{2 (8 x+1)} & 0 & \frac{x+1}{2 (8 x+1)} \\
\end{array}
\right),
\ee
\end{center}
where $0<x<1$.

\begin{figure}[t]
 \centering
  \includegraphics[width=8cm]{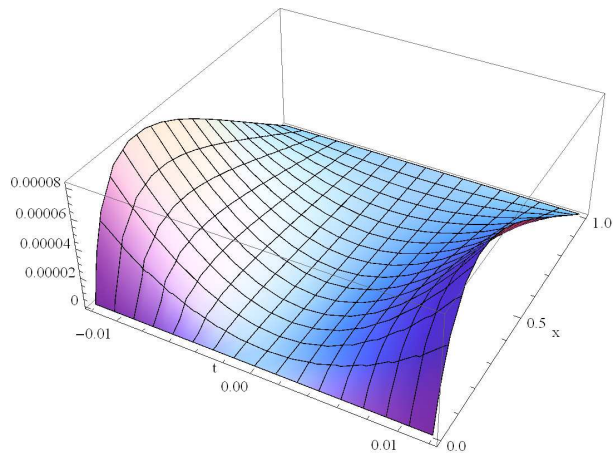}
  \caption{\label{fig:1} The value of $|\mathcal{P}\|_{\mathbf{tr}}-\frac{9a+1}{12}$ for the state $\rho^x$, where $x\in(0,1)$ and $t\in[-0.012,0.012]$.}
\end{figure}

By straightforward computation, we have that $\|\mathcal{P}\|_{\mathbf{tr}}>\frac{9a+1}{12}$ for $0<x<1$.
Thus, our criterion can detect the entanglement for the whole family of $3\times 3$ bound entangled states.
In Fig.~\ref{fig:1}, we plot the value of $|\mathcal{P}\|_{\mathbf{tr}}-\frac{9a+1}{12}$ as a function of $x$ and $t$.

Now we add white noise to $\rho^x$, and consider
\be
\rho(x, q)=q\rho^x+\frac{(1-q)}{9}\mathbb{I}, ~~ 0\leqslant q \leqslant 1.
\ee
Using the same general SIC-POVMs, we have
\be
J_a(\rho(x, q))-\frac{9a+1}{12}&&=\sum\limits_{j=1}\limits^{d^2}\Tr(P_j\otimes \bar{P}_j\rho(x, q))-\frac{9a+1}{12}\nonumber\\
 &&=24t^2(-4+\frac{q+35qx}{1+8x}).
\ee
From Fig.~\ref{fig:2}, one can easily find that $J_a(\rho(x, q))-\frac{9a+1}{12}<0$ for all permissible $x,\ q$. Thus, our criterion is shown to be more efficient in detecting entanglement of $\rho(x, q)$ than the criterion of Ref. \cite{33}.

\begin{figure}
 \centering
  \includegraphics[width=8cm]{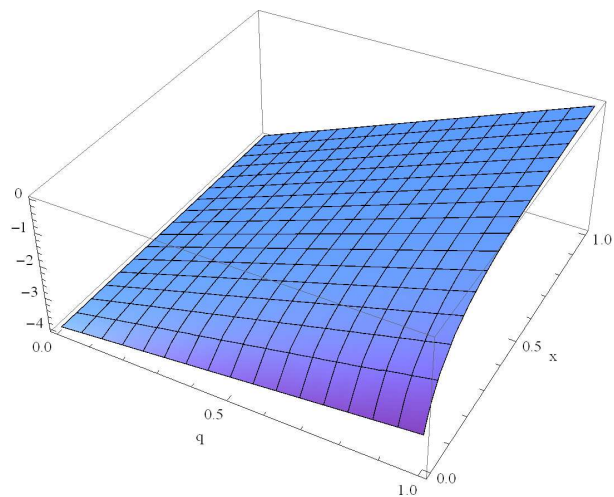}\\
  \caption{ \label{fig:2}The value of $-4+\frac{q+35qx}{1+8x}$ for $x\in(0,1)$ and $q\in[0,1]$.}
\end{figure}

Moreover, our criterion can successfully detect entanglement for some states that cannot be detected by the criterion in \cite{38}. Let us consider the following three states,
$\rho(0.25, 0.994)$, $\rho(0.45, 0.995)$ and $\rho(0.57, 0.996)$, whose entanglement cannot be identified by the criterion in \cite{38}. Denote the correlation entries  by  $[\mathcal{P}^{\alpha}]_{ij}=\langle P_i\otimes \bar{P}_j\rangle$, $i,j=1,\dots,9$, $\alpha=1,2,3$, for three states $\rho(0.25, 0.994)$, $\rho(0.45, 0.995)$ and $\rho(0.57, 0.996)$, respectively. We have $\|\mathcal{P}^1\|_{\mathbf{tr}}>\frac{9a+1}{12}$, $\|\mathcal{P}^2\|_{\mathbf{tr}}>\frac{9a+1}{12}$ and $\|\mathcal{P}^3\|_{\mathbf{tr}}>\frac{9a+1}{12}$ in the respective fixed parameter interval, see Fig.~\ref{fig:3}.  Thus, our criterion can successfully detect the entanglement of these states by suitably choosing the
general SIC measurements, namely, the parameter $t$. Therefore, in this case our criterion is more efficient than the criterion in \cite{38}.

\begin{figure}
\centering
\subfigure{
\begin{minipage}[b]{0.4\textwidth}
\includegraphics[width=1\textwidth]{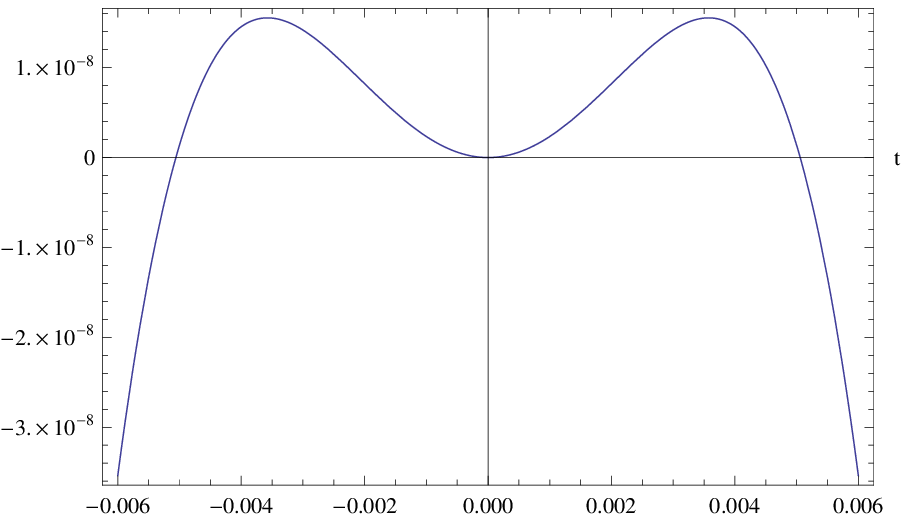}\\
\includegraphics[width=1\textwidth]{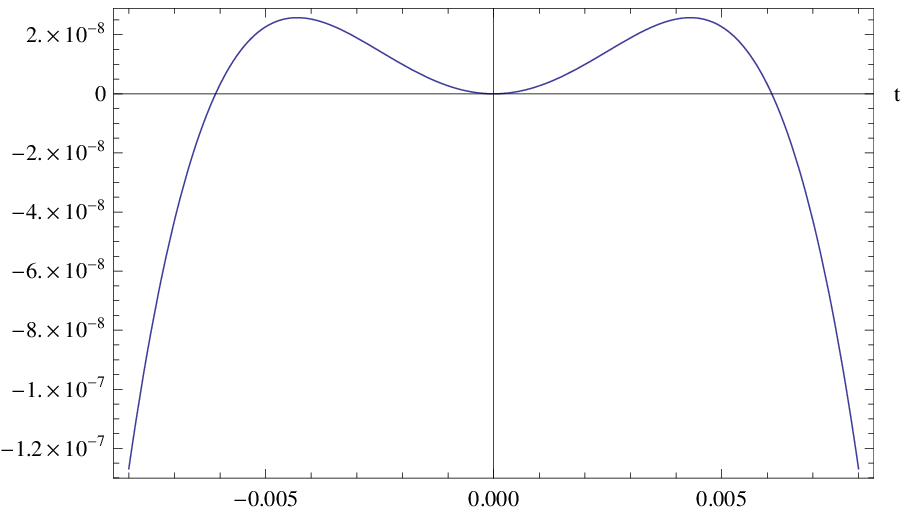}\\
\includegraphics[width=1\textwidth]{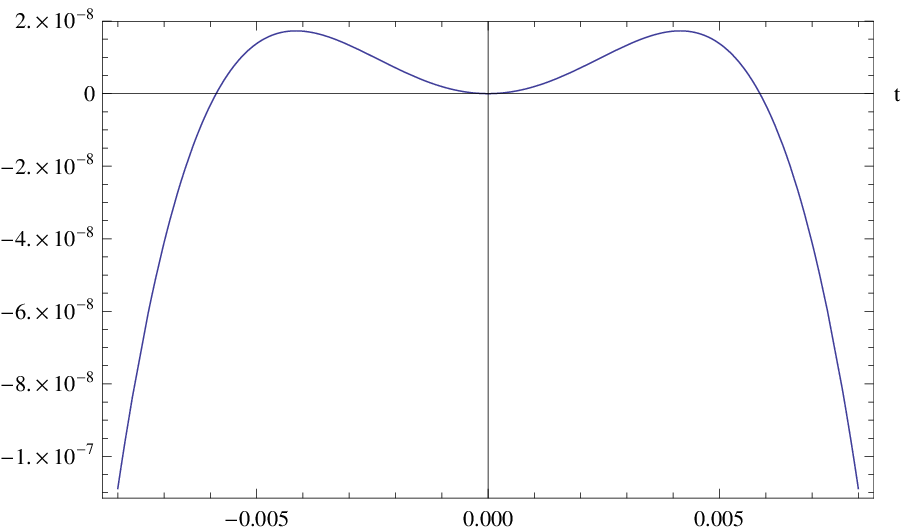}
\end{minipage}
}
\caption{\label{fig:3}The value $\| \mathcal{P}\|_{\mathbf{tr}}-\frac{9a+1}{12}$ for the  states $\rho(0.25, 0.994)$, $\rho(0.45, 0.995)$ and $\rho(0.57, 0.996)$ (form top to bottom).}
\end{figure}

\section{Conclusion}\label{sec:4}

We have presented an entanglement detection criterion constructed from the general SIC measurements.  Interestingly, this construction includes the criterion constructed by the SIC measurements as a special case that the parameter $a$ of general SIC measurements is equal to ${1}/{d^2}$. The criterion has been shown to be more efficient in detecting entanglement of some quantum states than the existing criteria.
Moreover, our separability criterion is experimentally feasible.

\begin{acknowledgements}
This work is supported by the NSF of China under Grant No.11675113, the Research Foundation for Youth Scholars of BTBU QNJJ2017-03, Beijing Municipal Commission of Education under Grant Nos. KM 201810011009 and KZ201810028042.
\end{acknowledgements}

\end{document}